\documentclass[preprint]{elsarticle}  
\usepackage{url}
\usepackage{amsmath}
\usepackage{amsfonts}
\usepackage{amsthm}

\usepackage{siunitx}
\usepackage{color}

\usepackage{graphics}
\usepackage{float}
\usepackage{tikz}
\usetikzlibrary{matrix,positioning,arrows,decorations.pathmorphing,shapes}
\usepackage{bbm}

\newtheorem{theorem}{Theorem}
\newtheorem{corollary}{Corollary}
\newtheorem{lemma}{Lemma}
\newtheorem{proposition}{Proposition}
\newtheorem{conjecture}{Conjecture}
\newtheorem{definition}{Definition}

\newcommand{\Glasses}[1]{
\begin{tikzpicture}[scale=4.0,>=stealth,shorten >=0.1pt, auto, semithick,
                     nodop/.style={circle,draw=black,fill=black,minimum size=3pt, inner sep=0pt},
                     every to/.style={draw,thin,black}]

\def \msqt {0.8660254037844386}
\def \triangleScale {1}
\def \angleAdjust {-30}

\centering

\begin{scope}[xshift = -0.5cm, yshift = -1.6cm, scale = 0.5]
\node at (0, -1) {#1};

\node [nodop] (1) at (-0.25,0) {};
\node [nodop] (11) at (0.25,0) {};
\node [nodop] (2) at (-0.25-\triangleScale * \msqt, 0.5 * \triangleScale) {};
\node [nodop] (3) at (-0.25-\triangleScale * \msqt, -0.5 * \triangleScale) {};
\node [nodop] (4) at (+0.25+\triangleScale * \msqt, 0.5 * \triangleScale) {};
\node [nodop] (5) at (+0.25+\triangleScale * \msqt, -0.5 * \triangleScale) {};

\path (1) edge (2) edge (3) edge (1) edge (11) ;
\path (11) edge (4) edge (5);

\path (2) edge (3);

\path  (4) edge (5);
\end{scope}
\end{tikzpicture}}

\makeatletter
\def\ps@pprintTitle{%
   \let\@oddhead\@empty
   \let\@evenhead\@empty
   \def\@oddfoot{\reset@font\hfil\thepage\hfil}
   \let\@evenfoot\@oddfoot
}
\makeatother

\begin{document}

\title{Analysis and Reliability of Separable Systems}





\author[1]{H\'ector Cancela} \ead{cancela@fing.edu.uy}
\author[2]{Gustavo Guerberoff\corref{cor1}} \ead{gguerber@fing.edu.uy}
\author[1,2]{Franco Robledo}\ead{frobledo@fing.edu.uy}
\author[1,2]{Pablo Romero}\ead{promero@fing.edu.uy}


\address[1]{Departamento de Investigaci\'on Operativa, Instituto de Computación, Facultad de Ingenier\'ia, Universidad de la República, Montevideo, Uruguay}
\address[2]{Laboratorio de Probabilidad y Estad\'istica, Instituto de Matemática y Estadística, Facultad de Ingenier\'ia, Universidad de la República, Montevideo, Uruguay}

\cortext[cor1]{Corresponding author}


\begin{abstract}
The operation of a system, such as a vehicle, communication network or automatic process, heavily depends on the correct operation of its components. A Stochastic Binary System (SBS) mathematically models the behavior of on-off systems, where the components are subject to probabilistic 
failures. Our goal is to understand the reliability of the global system. 

The reliability evaluation of an SBS belongs to the class of NP-Hard problems, and the combinatorics of SBS imposes several challenges. In a previous work by the same authors, a special sub-class of SBSs called \emph{separable systems} was introduced. These systems accept an efficient representation by a linear inequality on the binary states of the components. However, the reliability evaluation of separable systems is still hard.

A theoretical contribution in the understanding of separable systems is given. We fully characterize 
separable systems under the all-terminal reliability model, finding that they admit efficient reliability evaluation in this relevant context. 

\end{abstract}


\maketitle

\section{Introduction}

Research on system reliability includes different models, metrics and algorithms for analyzing how the possible failures of components affect the behaviour of a complex system. 
The practical applications of system reliability analysis are steadily growing in frequency and diversity. We mention just a few examples here. For instance, the paper~\cite{JOHANSSON201327} considers a joint reliability/vulnerability analysis of critical infrastructures, with a practical impact into electrical networks. The article~\cite{MACCHI201271} develops a reliability model in order to identify critical elements in a railway system. The results have practical use in the Italian public company
Rete Ferroviaria Italiana. The work by Li et al. \cite{Li201624} discusses two measures of infrastructure networks increasingly dependent on information systems, namely connectivity reliability  and the topological controllability in terms of topology, robustness, and node importance, taking  eight city-level power transmission networks and thousands of artificial networks to discuss the use of  these measures to improve reliability-based design. The work  \cite{MURIELVILLEGAS2016151} analyzes the impact of natural hazards on transportation networks in Colombia, focusing on the connectivity reliability and vulnerability of inter-urban transportation affecting remote populations in the case of disasters such as floods.

Classical network analysis was developed taking into account a system modelled as a graph, where either nodes or links (or both) are subject to failure, and where the system as a whole works correctly when the subgraph resulting from deleting the failed components verifies some connectivity constraint (in general, that a given subset of nodes is connected; this includes all terminal connectivity and source-terminal connectivity as special cases).
The theory of stochastic binary systems can be seen as a generalization of network reliability models; as the system structure function in an SBS can be any arbitrary boolean function of the components, instead of a variant of a connectivity function over a graph.

The mathematical understanding of SBS involves challenges in terms of complexity, combinatorics, and reliability analysis. In the most general setting, even the determination of operational or 
non-operational configurations are algorithmically hard problems. Under a realistic assumption 
of monotonicity or \emph{well-behavior}, finding minimally operational configurations 
accept efficient algorithms; see~\cite{Romero2019} for details. The interplay between static systems 
and dynamical stochastic binary systems is also being explored, in terms of stochastic processes~\cite{Lagos20}. The lifetimes of independent random variables govern a stochastic process where the components fail, until a non-operational system is obtained. An elegant interplay between these dynamical systems provides a new approach to reliability estimation.\\

The reliability evaluation of SBS belongs to the class of $\mathcal{NP}$-Hard problems. 
Furthermore, reliability evaluation of special (well-behaved) SBS also belongs to this class. This fact promotes the need of distinguished sub-classes, and the development of novel approximative techniques.  In~\cite{RNDM2018} the concept of separability in stochastic binary systems was introduced. 
As discussed below, separable systems are those whose structure function can be characterized by a  hyperplane separating operational from failure states. 
Separable systems have some particular properties which can be of interest in the study of system reliability.

This paper aims to advance in the analysis of separable system. 
Specifically, the contributions of the present work can be summarized as follows:
\begin{enumerate}
\item The $\mathcal{NP}$-Hardness of the reliability evaluation for separable systems is established.
\item A separable system under the all-terminal reliability model is called a separable graph. We fully characterize separable graphs. As a corollary, we conclude that the reliability evaluation of separable graphs can be obtained in linear time.
\item A discussion of the level of separability for non-separable systems is presented. 
\end{enumerate}

The remainder of this paper is organized in this way. Section~\ref{sbs} presents the main definitions of stochastic binary systems and separable systems. 
A theoretical analysis of separable systems is covered in Section~\ref{characterization}. It includes a characterization of separable systems, as well as the  analysis of the computational complexity  of the reliability 
evaluation for these systems. A particular analysis of the all-terminal reliability model is presented in Section~\ref{all-terminal}. Generalizations of the concept of separability by hyperplanes are studied in Section~\ref{dseparability}. Finally, Section~\ref{conclusiones} presents concluding remarks and trends for future work.

\section{Stochastic Binary Systems and Separable Systems}\label{sbs}
 
In this paper we will use the following definitions and notation.

\begin{definition}[SBS]
A \emph{stochastic binary system} (SBS) is a triad $\mathcal{S} = (S, r,\phi)$:
\begin{itemize}
\item $S$: ground (finite) set of \emph{components}, usually $S=\{1,\ldots,N\}$; a configuration or a \emph{state} of the system is an element of $\Omega=\{\sigma:S\to \{0,1\}\}$.
\item $r$: probability measure on  $\Omega$.
\item $\phi:\Omega \to \{0,1\}$: \emph{structure} function.
\end{itemize}
\end{definition}

Given a state $\sigma\in\Omega$, $\phi(\sigma)=1$ means that the  system is in an operational state; we call $\sigma$ a {\em pathset}. Respectively, if $\phi(\sigma)=0$ then  the  system is in a failure state; we call $\sigma$ a {\em cutset}.

The reliability of a SBS  is its probability of correct operation: 
\begin{equation}
 R_{\mathcal{S}} = P(\phi=1) = \sum_{\sigma\in\Omega:\phi(\sigma)=1}r(\sigma).
\end{equation}
The \emph{unreliability} of $\mathcal{S}$ is $ U_{\mathcal{S}}=1- R_{\mathcal{S}}$.

\begin{definition}[SMBS]
An SBS is {\em monotone} if the structure function $\phi$ is monotonically increasing with respect to the usual partial order in $\Omega$,  $\phi({\bold 0})=0$ and $\phi({\bold 1})=1$. We denote such an SBS as a Stochastic Monotone Binary System (SMBS).
\end{definition}

\begin{definition}[Minpaths/Mincuts]
Let $\mathcal{S} = (S, r,\phi)$ be an SMBS:
\begin{itemize}
\item A pathset $\sigma$ is a \emph{minpath} if $\phi(\omega)=0$ for all $\omega<\sigma$.
\item A cutset $\omega$ is a \emph{mincut} if $\phi(\sigma)=1$ for all $\sigma>\omega$.
\item A $\sigma$-ray is the set $S_\sigma=\{ \omega\in \Omega: \omega\geq \sigma\}$.
\end{itemize}
\end{definition}

An SMBS is fully characterized by its mincuts (or its minpaths). 
In fact, if we are given the complete list of minpaths, then the complete list of pathsets is precisely the union of the $\sigma$-rays among all minpaths $\sigma$.





As the class of SMBSs include the classical $K$-terminal graph reliability problem, which is known to belong to the $\mathcal{NP}$-Hard class~\cite{Rosenthal}, the  reliability evaluation of an SMBS belongs to the class of $\mathcal{NP}$-Hard problems. Of course the same applies for the (still more general) problem of reliability evaluation of an SBS.


We consider in what follows $S=\{1,\ldots,N\}$, so that $\Omega=\{0,1\}^N$ is the set of the extremal points of the unit hypercube  in $\mathbb{R}^{N}$.

A hyperplane  in the Euclidean space $\mathbb{R}^N$ is fully characterized by its normal vector $\vec{n}$ and a point $P$ that belongs to the hyperplane: $\langle \vec{n},X-P \rangle =0$, where
$\langle x,y \rangle =\sum_{i=1}^{N} x_iy_i$ is the inner product. If we denote $\vec{n}=(n_1,\ldots,n_N)$ and $\langle \vec{n},P \rangle =\alpha_0$, the points of the hyperplane are those satisfying the equation $\sum_{i=1}^{N} n_ix_i = \alpha_0$.
For ease of discussion (and without losing generality) we will choose the representation of a hyperplane so that any  cutset $\omega$ lies on the hyperplane or in its negative side (i.e. the geometric points that verify $\sum_{i=1}^{N} n_i\omega_i \leq \alpha_0$), and any pathset $\sigma$ lies on the positive side of the hyperplane (i.e. $\sum_{i=1}^{N} n_i\sigma_i >\alpha_0$).

Such representation is justified by the following observation: for any separating hyperplane $H$, there exists $H^{\prime} \sim H$ (that is, a hyperplane $H^{\prime}$ separating the same subset of cutsets and pathsets as $H$), 
with non-negative components of the normal vector, such that $\|\vec{n} \|_1 = \sum_{i=1}^{N}n_i = 1$. Also, it is possible to replace this normal vector by another non-negative normal vector with unit 1-norm.

 \begin{definition}[Separable System]
An SBS is \emph{separable} if the cutsets/pathsets can be separated by some hyperplane in $\mathbb{R}^{N}$.
\end{definition}

It turns out that separable systems are a special sub-class of SMBSs. However, in Section~\ref{characterization} we will build an infinite family of examples of SMBSs that are not separable. 

We enumerate a number of relevant results about separable SMBS presented in our previous conference paper \cite{RNDM2018}, which will be useful in the remainder of the paper.

\section{Analysis of Separable Systems}\label{characterization}
In this section, we first study the hardness of the reliability evaluation for separable systems. 
Then, we provide two alternative characterizations of these systems.

\subsection{Complexity}
Even though separable systems accept an efficient representation, their reliability evaluation is computationally hard~\cite{Romero2019}. The key is a reduction from {\it PARTITION}, which belongs 
to the class of $\mathcal{NP}$-Hard problems~\cite{Garey:1979:CIG:578533}. 
A slight variation of the proof offered in~\cite{Romero2019} is given here: 
\begin{theorem}\label{complexseparable}
The reliability evaluation of separable systems belongs to the class of $\mathcal{NP}$-Hard problems.
\end{theorem}

\begin{proof}
By reduction from {\it PARTITION}. Consider an instance of natural numbers $A=\{a_1,\ldots,a_N\}$, and let 
$s=\sum_{i=1}^{N}a_i$ be the sum over the elements of the list. 
Let us define $n_i = \frac{a_i}{s}$, $n_{min}=\min_{i=1,\ldots,N}\{n_i\}$, and consider the separable systems $\mathcal{S}_1$ and $\mathcal{S}_2$: 
\begin{enumerate}
\item $\mathcal{S}_1$ is characterized by the hyperplane $\sum_{i=1}^{N}n_i x_i =  \frac{1}{2}+\frac{n_{min}}{2}$; 
\item $\mathcal{S}_2$ is characterized by the hyperplane $\sum_{i=1}^{N}n_i x_i = \frac{1}{2}$; 
\end{enumerate}
Observe that the difference of the reliability of both systems evaluated at $p=1/2$ is:
\begin{align*}
R_{\mathcal{S}_2}(1/2)-R_{\mathcal{S}_1}(1/2) 
&=P(\sum_{i=1}^{N}n_i \sigma_i\geq\frac{1}{2})-P(\sum_{i=1}^{N}n_i \sigma_i\geq\frac{1}{2}+\frac{n_{min}}{2})\\
&= P(\sum_{i=1}^{N}n_i \sigma_i=\frac{1}{2})\\
&= \frac{\# \{(\sigma_1,\ldots,\sigma_N)\in \{0,1\}^N: \sum_{i=1}^{N}n_i \sigma_i=\frac{1}{2}\}}{2^N},
\end{align*}
and the last number is positive if and only if there exists a subset $B \subseteq \{1,\ldots,N\}$ such that 
$\sum_{i\in B}n_i=\frac{1}{2}$. In that case, if we multiply on both sides by $s$ we get that $\sum_{i \in B}a_i = \frac{s}{2}$, and the answer to {\it PARTITION} for the list $A$ is {\it YES}. Otherwise, the answer to {\it PARTITION} is {\it NO}. 
Therefore, the reliability evaluation of separable systems is at least as hard as {\it PARTITION}, and it belongs to the class of $\mathcal{NP}$-Hard problems. 
\end{proof}

\subsection{Characterizations}
Separable systems can be characterized using Hahn-Banach separation theorem for compact and convex sets~\cite{rudin1974functional}. If $CH(\mathcal{P})$ and $CH(\mathcal{C})$ denote the convex hull of the pathsets and cutsets respectively, then following result holds:
\begin{theorem}[Prop. 3,~\cite{RNDM2018}]\label{hahn}
An SBS is separable iff $CH(\mathcal{P}) \cap CH(\mathcal{C}) = \emptyset$.
\end{theorem}
\begin{proof}
If the intersection is empty, the Hahn-Banach separation theorem for convex sets asserts that there exists a hyperplane $H$ that separates those convex sets; see Theorem 3.4 in~\cite{rudin1974functional}  for a full proof. As a consequence, $\phi = \phi_H$ for some hyperplane $H$. 

For the converse, we know that the SBS is separable. Therefore, there exists some hyperplane $H) \sum_{i=1}^{N}n_ix_i=\alpha_0$ such that  $\sum_{i=1}^{N}n_i \omega_i \leq \alpha_0$ for cutsets, and $\sum_{i=1}^{N}n_i \sigma_i > \alpha_0$ for pathsets. Suppose for a moment that $CH(\mathcal{P}) \cap CH(\mathcal{C}) \neq \emptyset$. 
There exists some element $z \in \mathbb{R}^{N}$ such that:
\begin{equation}
z = \sum_{j=1}^{l}\alpha_j\sigma^j = \sum_{k=1}^{s}\beta_k \omega^k,
\end{equation}
for some states $\sigma^1,\ldots,\sigma^l \in \mathcal{P}$, $\omega^1,\ldots,\omega^s \in \mathcal{C}$, and non-negative numbers such that $\sum_{j=1}^{l}\alpha_j = \sum_{k=1}^{s}\beta_k=1$. If we denote $\sigma^j=(\sigma^j_{1},\ldots,\sigma^j_{N})$ we know that $\sum_{i=1}^{N}n_i \sigma^j_{i} > \alpha_0$.
Therefore, for $z=(z_1,\ldots,z_N)$ we get that:
\begin{align*}
\sum_{i=1}^{N}n_iz_i &= \sum_{i=1}^{N}n_i (\sum_{j=1}^{l}\alpha_j\sigma^j_{i})\\
                     &= \sum_{j=1}^{l}\alpha_j [\sum_{i=1}^{N}n_i\sigma^j_{i}]\\
                     &> (\sum_{j=1}^{l}\alpha_j)\alpha_0 = \alpha_0.
\end{align*}
Analogously, using the fact that $z = \sum_{k=1}^{s}\beta_k\omega^k$ we get that $\sum_{i=1}^{N}n_iz_i \leq \alpha_0$, which is a contradiction.
Therefore we must have $CH(\mathcal{P}) \cap CH(\mathcal{C}) = \emptyset$, and the result holds.
\end{proof}

While the structure function of some SMBS can be defined by a hyperplane, there exist SMBSs that are not separable. In fact, consider an arbitrary number of components $N \geq 4$, 
and the SMBSs family $\mathcal{S}_N$ characterized by two mincuts $\{\omega^1,\omega^2\}$ such that 
$\omega^2 = {\bold 1} - \omega^1$, and $\omega^1$ is defined as follows: 
\begin{align*}
\omega^1_i &=1, \, \,  \forall i=1,\ldots,\lfloor N/2 \rfloor,\\
\omega^1_i &=0, \, \,  \forall i=\lfloor N/2 \rfloor+1,\ldots, N.
\end{align*}

Consider the states $\sigma^1$ and $\sigma^2$, such that $\sigma^2 = {\bold 1} - \sigma^1$, and 
\begin{align*}
\sigma^1_{2i-1} &= 1, \, \forall i=1,\ldots, \lfloor N/2 \rfloor,\\
\sigma^1_{2i} &= 0, \, \forall i=1, \ldots, \lfloor N/2 \rfloor.
\end{align*}
First, observe that $\sigma^1$ and $\sigma^2$ are not upper-bounded by $\omega^1$ nor $\omega^2$ 
(i.e., the inequalities $\sigma^i \leq \omega^j$ do not hold, for any pair $i,j \in \{1,2\}$). Therefore, $\sigma^1$ and $\sigma^2$ must be pathsets, since $\mathcal{S}_N$ is the SBS characterized by the mincuts $\omega^1$ and $\omega^2$. 
Further, $\frac{\sigma^1+\sigma^2}{2}=\frac{\omega^1+\omega^2}{2} = \frac{1}{2} {\bold 1}$. 
By Hahn-Banach theorem, the infinite family of SMBSs $\{\mathcal{S}_N\}_{N \geq 4}$ is not separable (see Theorem~\ref{hahn}). 
By an exhaustive analysis, it can be observed that all SMBS are separable systems when the 
number of components is not greater than three. A geometric interpretation is also feasible in this cases. However, the challenges arise in higher dimensions.\\


In the following, we consider an alternative characterization of separable systems in terms 
of weighted cutsets and pathsets. Consider an arbitrary assignment $n_1,\ldots,n_N$ 
of non-negative numbers to the respective components of the system. The condition 
$\sum_{i:\sigma_i=1}n_i \geq \alpha_0$ for all the pathsets is equivalent to finding 
the pathset $\sigma$ with minimum-cost, $c(\sigma)=\sum_{i:\sigma_i=1}n_i$, and testing if $c(\sigma) \geq \alpha_0$. 
Analogously, the condition $\sum_{i:\omega_i=1}n_i < \alpha_0$ for all the cutsets is equivalent to testing whether the cutset $\omega$ with minimum cost, $c(\omega)=\sum_{i:\omega_i=0}n_i$, 
satisfies the test $S-c(\omega) < \alpha_0$, where $S=\sum_{i=1}^{N}n_i$ is the cost of the global system. 
Observe that, for convenience, the cost of a cutset is defined as the sum of the components under failure. In particular, we get the following characterization of separable systems:
\begin{theorem}\label{assign}
An SBS is separable if and only if there exists an assignment of non-negative costs to the components $\{n_i\}_{i=1,\ldots,N}$ such that $S<c(\sigma)+c(\omega)$, where $c(\sigma)$ and $c(\omega)$ denote pathset/cutset minimum-cost respectively. 
\end{theorem}
\begin{proof}
First, let us assume that we have a separable SBS with hyperplane $\sum_{i=1}^{N}n_ix_i = \alpha_0$. 
Using the previous reasoning, the assignment $\{n_i\}_{i=1,\ldots,N}$ verifies $c(\sigma) \geq \alpha_0$ 
and $S-c(\omega) < \alpha_0$. Therefore, $S < c(\omega)+c(\sigma)$.

For the converse, let us fix $\alpha_0=c(\sigma)$, the pathset with minimum cost. Clearly, the specific 
pathset $\sigma$ meets the condition $\sum_{i=1}^{N}n_i\sigma_i \geq \alpha_0$; in fact the equality is met. 
By its definition, the inequality holds for the other pathsets. Finally, 
we use the fact that $S<c(\omega)+c(\sigma)$ to verify that the cutset with minimum-cost, $\omega$, meets the inequality 
$\sum_{i=1}^{N}n_i\omega_i < \alpha_0$. The inequality for the other cutsets is straight since $\omega$ 
is a cutset with minimum-cost. Therefore, the SBS is separable, concluding the proof.
 \end{proof}

\section{Separability in Graphs} \label{all-terminal}
Our characterization of separable systems has a straightforward reading in the all-terminal 
reliability model. 
\begin{definition}[Separable Graph]
A graph $G=(V,E)$ is \emph{separable} if there exists an assignment of non-negative real numbers $n_1,\ldots,n_m$ to its $m$ links, and there exists a threshold $\alpha$ such that $c(E^{\prime}) \geq \alpha$ if and only if the spanning subgraph $G^{\prime}=(V,E^{\prime})$ is connected. 
\end{definition}
Let $G$ be a connected weighted graph. Recall the Kruskal algorithm provides efficiently the cost of the minimum spanning tree, $MST(G)$. Furthermore, the cutset with minimum-cost, $m(G)$, is obtained using the Ford-Fulkerson algorithm. 
Therefore, the following corollary of Theorem~\ref{assign} holds for graphs:
\begin{corollary}
A graph is separable iff there exists a feasible assignment $\{n_i\}_{i=1,\ldots,N}$ to the links such that $S<MST(G)+m(G)$, where $MST(G)$ is the cost of the minimum spanning tree, $m(G)$ the mincut with minimum capacity, 
and $S=\sum_{i=1}^{N}n_i$ the sum of the link weights. 
\end{corollary}
For example, trees and elementary cycles are separable graphs. Indeed, if $T_n$ is a tree with $n$ nodes, 
a feasible assignment is an identical unit-cost to all the links, since in that case $MST(T_n)=n-1$, 
$m(G)=1$ and the global sum is $S=n-1 < MST(T_n)+m(T_n)$. Analogously, if $C_n$ denotes the elementary cycle with $n$ nodes, 
then $S=n<(n-1)+2=MST(C_n)+m(C_n)$, and the same unit-cost assignment works.

Intuitively, if the graph is dense enough, 
one would expect that the combined cost of a minimum spanning tree and mincut should not exceed $S$, the global cost of the graph.

Our first result deals with the extremal case of complete graphs:

\begin{proposition}
Complete graphs $(K_n)_{n \geq 4}$ are nonseparable.
\end{proposition}
\begin{proof}
Consider an arbitrary assignment $\{n_i\}_{i=1,\ldots,n(n-1)/2}$ to the links of $K_n$, and an arbitrary star-graph $K_{1,n}$ contained in $K_n$. Since $K_{1,n}$ is connected, its cost is greater than, or 
equal to the cost of the minimum spanning tree, so, $c(K_{1,n}) \geq MST(K_n)$. Furthermore, the complementary links of 
$K_{1,n}$, or the complementary graph $K_{1,n}^C$, is a cutset (it isolates a single node), so 
the cost must exceed the mincut: $c(K_{1,n}^C) \geq m(K_n)$. 
But then, the global cost is $c(K_n)=c(K_{1,n})+c(K_{1,n}^C) \geq MST(K_n)+m(K_n)$. 
The conclusion is that $S=c(K_n) \geq MST(K_n)+m(K_n)$ for any feasible assignment, and $K_n$ is nonseparable.  
\end{proof}

With the following lemmas, we will present a hereditary property of separable graphs, stated in Theorem~\ref{hereditary}.  
Consider a simple connected graph $G=(V,E)$. We will consider two different link additions:
\begin{itemize}
\item We denote $G_{in}=G+e_{in}$ to the resulting graph after the addition of an internal link $e_{in}=\{u_1,u_2\}$, where $u_1,u_2\in V$.
\item We denote $G_{out}=G+e_{out}$ to the resulting graph after the addition of an external link $e_{out}=\{u_1,u_2\}$, where $u_1\in V$ but $u_2 \notin V$.
\end{itemize}
Observe that $G+e_{in}$ and $G$ share an identical node-set $V$, while the node-set for $G+e_{out}$ is $V \cup \{u_2\}$. 

\begin{lemma} \label{out}
If $G$ is nonseparable then $G_{out}$ is nonseparable.
\end{lemma}
\begin{proof}
Suppose for a moment that there exists a feasible assignment $\{n_i\}_{i=1,\ldots,N+1}$ for $G_{out}$. 
Then:
\begin{align*}
(\sum_{i=1}^{N}n_i)+n_{N+1} &< MST(G_{out})+m(G_{out})\\
                            &= MST(G)+n_{N+1}+\min\{m(G),n_{N+1}\}\\
                            &\leq MST(G)+n_{N+1}+m(G),
\end{align*}
and $\{n_i\}_{i=1,\ldots,N}$ would be a feasible assignment for $G$, which is a contradiction. 
Therefore, $G_{out}$ is nonseparable.
\end{proof}

\begin{lemma} \label{in}
If $G$ is nonseparable then $G_{in}$ is nonseparable.
\end{lemma}
\begin{proof}
Suppose for a moment that there exists a feasible assignment $\{n_i\}_{i=1,\ldots,N+1}$ for $G_{in}$. Then:
\begin{align*}
(\sum_{i=1}^{N}n_i)+n_{N+1} &< MST(G_{in}) + m(G_{in})\\
                                      &\leq MST(G)+m(G)+n_{N+1},
\end{align*}
and $\{n_i\}_{i=1,\ldots,N}$ would be a feasible assignment for $G$, which is a contradiction. Therefore, $G_{in}$ is nonseparable. 
\end{proof}
Observe that Lemma~\ref{in} informally states that graphs with more density are nonseparable. Using the contrapositive of 
Lemmas~\ref{out}~and~\ref{in} we obtain the following:
\begin{theorem}\label{hereditary}
Separability is a hereditary property in graphs.
\end{theorem} 
\begin{proof}
Reading the contrapositive of Lemma~\ref{in}, we know that the deletion of one or several links from a separable graph 
is also separable. By Lemma~\ref{out}, we also know that a node-deletion in a separable graph (with the intermediate deletion of links using Lemma~\ref{in}) is also separable. Combining node and link deletions, an arbitrary subgraph is obtained, 
and it must be separable as well. 
\end{proof}

\begin{lemma}\label{add}
If $G$ is separable, $G_{out}$ is also separable.
\end{lemma}
\begin{proof}
Consider a feasible assignment $\{n_i\}_{i=1,\ldots,N}$ for $G$, where $S<MST(G)+m(G)$ holds. Let us consider an extended 
assignment with $n_{N+1}$ for the external link, such that $n_{N+1}>m(G)$. Then:
\begin{align*}
S+n_{N+1} &< (MST(G)+n_{N+1})+m(G)\\
               &=  MST(G_{out}) + \min\{m(G),n_{N+1}\}\\
               &=  MST(G_{out})+m(G_{out}),
\end{align*}
and $\{n_i\}_{i=1,\ldots,N+1}$ is a feasible assignment for $G_{out}$. 
\end{proof}

\begin{corollary}\label{arborescencia}
Cycles with arborescences are separable graphs
\end{corollary}
\begin{proof}
We know that elementary cycles are separable. The result follows by the addition of one or several trees hanging to 
different nodes from the first cycle. Supported by Lemma~\ref{add}, the separability is preserved by the addition of those links.
\end{proof}

Figure~\ref{ejemploMonma} depicts Monma graphs. These graphs have two degree-3 nodes connected by 3 node-disjoint paths. Every proper subgraph of a Monma graph is either a unicyclic graph, a tree, or a disconnected graph. Therefore, every proper subgraph of a Monma graph is separable. We will see that Monma graphs are minimally nonseparable graphs. 
 Clyde Monma et al. used these graphs to design  minimum cost biconnected metric networks~\cite{Monma1990}. Some (but not all) of these graphs also attain the maximum reliability among all the graphs with $p$ nodes and $q=p+1$ 
links~\cite{Boesch1991}.

\begin{figure}[htb!]\centering{
 \begin{tikzpicture}  [scale=1.5,>=stealth,shorten >=0.1pt, auto, semithick,
                     nodot/.style={circle,draw=black,minimum size=10pt, inner sep=0pt, font=\scriptsize, fill=gray!50},                     
                     nodop/.style={circle,draw=black,minimum size=10pt, inner sep=0pt, font=\scriptsize},
                     every to/.style={draw,thin,black}]
                     
  \begin{scope}[xshift=0cm,scale=1]

\node [nodop] (s) at (1,1) {$u$};
\node [nodop] (t) at (5,1) {$v$};
\node [nodot] (n02) at (2,0) {$c_1$};
\node [nodot] (n03) at (3,0) {$c_2$};
\node [nodot] (n04) at (4,0) {$c_{l_3}$};
\node [nodot] (n12) at (2,1) {$b_1$};
\node [nodot] (n13) at (3,1) {$b_2$};
\node [nodot] (n14) at (4,1) {$b_{l_2}$};
\node [nodot] (n22) at (2,2) {$a_1$};
\node [nodot] (n23) at (3,2) {$a_2$};
\node [nodot] (n24) at (4,2) {$a_{l_1}$};
\path (s) edge (n02);
\path (s) edge (n12);
\path (s) edge (n22);
\path (t) edge (n04);
\path (t) edge (n14);
\path (t) edge (n24);
\path[dotted] (n03) edge (n04);
\path[dotted] (n13) edge (n14);
\path[dotted] (n23) edge (n24);
\path (n02) edge (n03);
\path (n12) edge (n13);
\path (n22) edge (n23);
\end{scope}
\end{tikzpicture}}\caption{Monma graph $M_{l_1+1,l_2+1,l_3+1}$.}\label{ejemploMonma}
\end{figure}
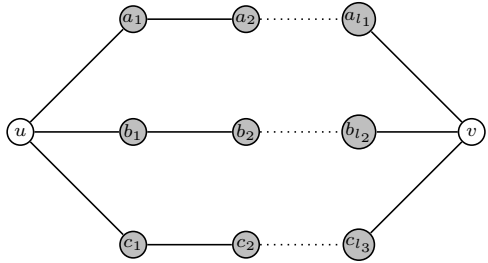

\begin{lemma}\label{monma}
Monma graphs are nonseparable
\end{lemma}
\begin{proof} 
Consider an arbitrary order for the links of Monma graph, and the rule $\phi(\sigma)=1$  
iff the Monma subgraph given by the links in subgraph $\sigma$ is connected. 
We will show that the convex hull of pathsets and cutsets meet at some point, and the result is established by Theorem~\ref{hahn}. 
Consider the four links $e_1=\{u,a_1\}$, $e_2=\{a_1,a_2\}$, $e_3=\{u,b_1\}$ and $e_4=\{b_1,b_2\}$ from Figure~\ref{ejemploMonma}. 
Let $1_{e_i,e_j}$ denote the binary word that is set to $1$ in all the bits but $0$ in the positions corresponding to 
the links $e_i$ and $e_j$. Consider the following identity:
\begin{equation}
\frac{1}{2}(1_{e_1,e_2}+1_{e_3,e_4}) = \frac{1}{4}(1_{e_1,e_3}+1_{e_1,e_4}+1_{e_2,e_3}+1_{e_2,e_4})
\end{equation} 
On one hand, we have a convex combination of cutsets. On the other, a convex combination of pathsets. 
By Theorem~\ref{hahn}, Monma graphs are nonseparable.
\end{proof}

Recall that a node $v$ in a graph $G$ is a cut-point if $G-v$ has more components than $G$. 
A connected graph is biconnected if it has no cut-points. The addition of an ear in a graph $G$ is the addition of an external elementary path between two different nodes from $G$. Whitney characterization theorem for biconnected graphs asserts that there exists an  ear decomposition of all biconnected graphs, such that $G=C_{s} \cup H_{1} \cup H_{2} \cup \dots \cup H_{r}$, 
$C_s$ is an elementary cycle and $H_{i}$ is the addition of an ear to the previous graph~\cite{Whitney}. A proof using modern terminology is given in the classical graph theoretical book~\cite{diestel2006graph}. This structural characterization 
of biconnected graphs leads us immediately to the following:
\begin{theorem}\label{jaja1}
Biconnected graphs are nonseparable, except for elementary cycles.
\end{theorem}
\begin{proof}
As the base-step, we know by Lemma~\ref{monma} that Monma graphs are nonseparable. If $G$ is biconnected and it is 
not an elementary cycle, then it has the addition of at least one ear of a cycle. Therefore, it has Monma as a subgraph. 
Therefore, Theorem~\ref{hereditary} asserts that $G$ cannot be separable. 
\end{proof}

Recall that the link-connectivity of a graph $G$ is the least number of links that must be removed in order 
to disconnect $G$. The Butterfly-graph consists of two triangles meeting in a common point (see Figure~\ref{Butterfly}).  
This is the smallest graph with link connectivity 2 that is not biconnected, since the kissing-point is a cut-point. 
As a consequence, it is natural to decide the separability of this graph:
\begin{lemma}
The Butterfly-graph $B$ is nonseparable
\end{lemma}
\begin{proof}
Consider an arbitrary assignment $\{n_i\}_{i=1,\ldots,6}$ for the links. 
We consider an assignment $n_1\leq n_2 \leq n_3$ in the left triangle, and $n_4\leq n_5 \leq n_6$ in the right triangle. Therefore $MST(B)=S-n_3-n_6$, 
and $m(B)=\min\{n_1+n_2,n_4+n_5\} \leq \min\{2n_2,2n_5\}\leq n_3+n_6$. 
This implies that $MST(B)+m(B) \leq S$ for all possible assignments in $B$, and $B$ has no feasible assignment.   
\end{proof}

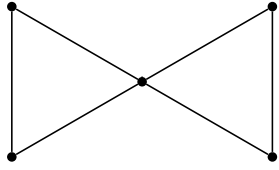
\begin{figure}
\centering
\begin{tikzpicture}[scale=4.0,>=stealth,shorten >=0.1pt, auto, semithick,
                     nodot/.style={circle,draw=black,fill=black,minimum size=3pt, inner sep=0pt, fill=gray!50},                     
                     nodop/.style={circle,draw=black,fill=black,minimum size=3pt, inner sep=0pt},
                     every to/.style={draw,thin,black}]

\def \msqt {0.8660254037844386}
\def \triangleScale {1}
\def \angleAdjust {-30}

\centering

\begin{scope}[ xshift = -2cm, yshift = 0cm, scale = 0.5]
\node at (0, -\triangleScale * 1) {};

\node [nodop] (1) at (0,0) {};
\node [nodop] (2) at (-\triangleScale * \msqt, 0.5 * \triangleScale) {};
\node [nodop] (3) at (-\triangleScale * \msqt, -0.5 * \triangleScale) {};
\node [nodop] (4) at (\triangleScale * \msqt, 0.5 * \triangleScale) {};
\node [nodop] (5) at (\triangleScale * \msqt, -0.5 * \triangleScale) {};

\path (1) edge (2) edge (3) edge (1) edge (4) edge (5) edge (1);

\path (2) edge (3);

\path  (4) edge (5);
\end{scope}

\end{tikzpicture}
\caption{Butterfly-graph $B$\label{Butterfly}.}
\end{figure}

An analogous reasoning leads to the following generalization:
\begin{corollary}\label{kiss}
Two kissing cycles are nonseparable.
\end{corollary}

A further generalization recalls Whitney characterization for bridgeless graphs: $G$ is a bridgeless graph if and only if $G=C_{s} \cup H_{1} \cup H_{2} \cup \dots \cup H_{r}$, 
$C_s$ is an elementary cycle and $H_{i}$ is the addition of an ear or a kissing cycle to the previous graph. The following result is analogous to Theorem~\ref{jaja1}:
\begin{corollary}\label{jaja2}
Bridgeless graphs are nonseparable, except for elementary cycles.
\end{corollary}

In order to fully characterize separable graphs, we need to study graphs that have at least one bridge $e \in G$. 
We already know that all the links in a tree are bridges, and they are separable graphs. Furthermore, 
cycles with arborescences are separable as well. Let us proceed our analysis with two triangles linked by a single bridge $e$, 
a graph called the Glasses-graph $B_e$.
\begin{lemma}
The Glasses-graph $B_e$ is nonseparable.
\end{lemma}
\begin{proof}
The reasoning is identical to the Butterfly-graph. Consider an assignment $\{n_i\}_{i=1,\ldots,7}$ as in the Butterfly-graph, 
but $n_7$ is the assignment for the bridge $e$. Therefore:
\begin{align*}
MST(B_e)+m(B_e)&=(S-n_3-n_6)\\ &+\min\{n_1+n_2,n_4+n_5,n_7\}\\ &\leq S,
\end{align*}
since $\min\{n_1+n_2,n_4+n_5,n_7\} \leq n_3 +n_6$, 
and the last inequality was already proved for the Butterfly-graph. 
\end{proof}

\begin{figure}
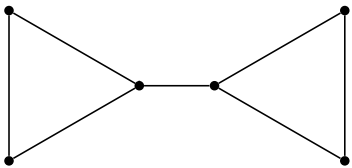

\centering
\Glasses{}
\caption{Glasses-graph $B_e$.}
\end{figure}

A slight generalization is possible:
\begin{corollary}\label{dosciclos}
Two cycles linked by an elementary path are nonseparable.
\end{corollary}

We are now ready to fully characterize separable graphs:
\begin{theorem}\label{characterize}
A graph $G$ is separable iff $G$ falls into one of the four categories:
\begin{enumerate}
\item $G$ is not connected;
\item $G$ is a tree;
\item $G$ is an elementary cycle;
\item $G$ is an elementary cycle with arborescences. 
\end{enumerate}
\end{theorem}
\begin{proof}
The proof of the reverse direction is easy, as all graphs in the four categories are separable:
\begin{enumerate}
\item If $G$ is disconnected, all of its configurations are cutsets and the reliability is null. In this case, 
the inequality $\sum_{i=1}^{N}\sigma_i > 2N$ is not satisfied by any binary vector $\sigma=(\sigma_1,\ldots,\sigma_N)$, and the 
graph is separable. 
\item If $G$ is a tree $T_N$ with $N$ links, the evidence is the hyperplane $\sum_{i=1}^{N}\sigma_i \geq N$. 
\item If $G=C_N$ is an elementary cycle, the evidence is the inequality $\sum_{i=1}^{N}\sigma_i \geq N-1$.
\item If $G$ is a tree with arborescences, Lemma~\ref{arborescencia} states that $G$ is separable. 
\end{enumerate}
To prove the direct direction, let $G$ be a separable graph, and assume $G$ is connected. 
We know by Corollary~\ref{jaja2} that $G$ must have a bridge. 
Combining Theorem~\ref{hereditary}~and~Corollary~\ref{jaja1}, we know that every  subgraph of $G$ must be an elementary cycle. Combining Corollaries~\ref{kiss}~and~\ref{dosciclos}, $G$ cannot have two cycles (either they are kissing or connected by a path). Therefore, $G$ is either a tree, an elementary cycle or an elementary cycle with arborescences. 
\end{proof}

\begin{corollary}\label{base}
The all-terminal reliability evaluation of separable graphs belongs to the class $\mathcal{P}$ of polynomial-time problems.
\end{corollary}
\begin{proof}
The analysis is straightforward. Let $G$ be a separable graph:
\begin{enumerate}
\item If $G$ is not connected, then $R(G)=0$.
\item If $G=T_N$ is a tree with $N$ links with independent reliabilities $(p_e)_{e\in T_N}$, then 
$R(G)=\prod_{e\in T_N}p_e$.
\item If $G=C_N$, then $$R(C_N)=\prod_{e\in C_N}p_e + \sum_{e\in C_N}(1-p_e)\prod_{e^{\prime}\neq e}p_{e^{\prime}}.$$
\item Finally, if $G$ is an elementary cycle with arborescences: $G=C_l \cup T_s$, being $T_s$ union of trees pending 
from the cycle $C_l$. Therefore, $R(G)=R(C_l)\times \prod_{e\in T_s}p_e$.
\end{enumerate}
The reader can appreciate that the reliability computation is a product, or a sum of products of the elementary link reliabilities. Therefore, the number of operations involved are linear, or quadratic, in the number of 
links.
\end{proof}

The corank of a graph is the number of independent cycles. In a connected graph with $n$ nodes and $m$ links, 
its corank is precisely $c(G)=m-n+1$. It is worth to remark that Theorem~\ref{characterize} can be re-stated in terms of corank: 
a connected graph $G$ is separable if and only if its corank is either $0$ or $1$.

We close this section by discussing a connection between the combinatorial optimization problem called the Network Utility Problem (NUP) and 
separable graphs. First, observe that an arbitrary spanning tree of a connected graph $G$ has $n-1$ links. Therefore, 
the corank of a graph is precisely the number of \emph{additional links that we must pay} to build the graph $G$, 
starting from a minimally-connected graph. In terms of communication, the corank of $G$ represents \emph{redundancy}. 
At the cost of redundancy, the resulting network can be robust under a certain amount of link failures. The profit is 
the link connectivity $\lambda(G)$, which represents the lowest number of links that should be removed in order to 
disconnect $G$. As a consequence, the \emph{utility} of a graph, $u(G)$, is the difference between the connectivity and 
the corank: $u(G)=\lambda(G)-c(G)=\lambda - m+n-1$. In~\cite{Canale2016}, the authors formally proved the following
\begin{theorem}
The graphs with maximum utility are exactly the trees and cycles. Their utility value is $1$.
\end{theorem}

\begin{corollary}
All the graphs with maximum utility are separable graphs.
\end{corollary}

The all-terminal reliability polynomial under identical elementary reliabilities in the links $r$ is 
\begin{equation}
R_G(r) = \sum_{i=\lambda(G)}^{c(G)-1}n_i(G)p^{m-i}(1-p)^{i} + \tau(G)p^{n-1}(1-p)^{m-n+1},
\end{equation}
where $n_i(G)$ is the number of connected subgraphs of $G$ with precisely $m-i$ links, and $\tau(G)$ is the tree-number of $G$, 
which can be found using Kirchhoff's Matrix-Tree  theorem~\cite{biggs1993algebraic}. Therefore, the number of unknowns is 
precisely the number of terms involved in the summation: $c(G)-\lambda$. The only cases where there are no terms in the sum 
occur either when $c(G)-\lambda=-1$, exactly in trees and cycles, or when $c(G)-\lambda=0$, only in an elementary cycle with arborescence, $K_4$, the Kite-graph and the Butterfly-graph~\cite{Canale2016}. These graphs are considered as the simplest in terms of reliability analysis. 
Indeed, in~\cite{Canale2016} the authors define the \emph{level of difficulty} of a graph as the difference 
$d(G)=c(G)-\lambda-1$, and a graph is easy if and only if $d(G)\leq 0$:
\begin{corollary}
All separable graphs are easy graphs.
\end{corollary}  
The reader can observe that the graphs with maximum utility $u(G)$ are the easiest graphs, with the minimum level of difficulty $d(G)$.

\section{$d$-Separability}
\label{dseparability}
A natural extension of our prior analysis is a classification of nonseparable systems. 

Let $\mathcal{S}= (S, r,\phi)$ be an arbitrary SMBS, and consider its corresponding $0$-$1$ 
labels of the vertices of a hypercube $Q_N$ in the Euclidean space $\mathbb{R}^N$. 

\begin{definition}[Level of Separability]\label{level}
The \emph{level of separability of $\mathcal{S}$} is the least positive integer $d$ such that 
occurs one of the following conditions:
\begin{itemize}
\item there exist $d$ hyperplanes such that all the pathsets reside in the intersection of the $d-$half spaces specified by the
non-negative normal components of the hyperplanes; or
\item there exist $d$ hyperplanes such that all the cutsets reside
in the intersection of the $d-$half spaces specified by the
opposite of the non-negative normal components.
\end{itemize}
\end{definition}

\begin{proposition}\label{vale}
Let $\mathcal{S}$ be an arbitrary SMBS, and let $\mu=|\mathcal{MC}|$ be the number of all its mincuts of $\mathcal S$. Then the level of separability d is at most $\mu$. 
\end{proposition}
\begin{proof}
Suppose that $\omega^{1},\ldots,\omega^{\mu}$ is the list of all the mincuts of $\mathcal{S}$. 
Consider the sets $S_i=\{j: \omega_{j}^{i}=0\}$, that represent the non-operational states 
for the mincut $\omega^{i}$. Observe that the mincut $\omega^{i}$ does not meet the inequality 
$\pi_i: \sum_{j\in S_i}\omega_j^i \geq 1$. Furthermore, the hyperplanes $\pi_1,\ldots,\pi_{\mu}$ 
meet the definition~\ref{level}, and the result follows. 
\end{proof}

Proposition~\ref{vale} shows that the level of separability will always be well defined 
for any arbitrary SMBS, thus it is 
an alternative way to classify a notion of \emph{difficulty} in the reliability evaluation for SMBSs. 

If we return to the all-terminal reliability model, we know all the graphs with level of separability $d=1$ (i.e., all separable graphs). From Theorem~\ref{characterize} we can observe that the Butterfly-graph, the Glasses-graph and Monma represent minimally nonseparable cases. For a better understanding of definition~\ref{level}, we find the level of separability in these minimally nonseparable cases in the following paragraphs. 

Let us denote $x_1,x_2,x_3$ and $y_1,y_2,y_3$ the states of the links for the Butterfly-graph, corresponding 
to both triangles (see Figure~\ref{Butterfly}). All pathsets must have at least two links from every triangle, 
and the following 2 hyperplanes determine pathsets:
\begin{align*}
x_1+x_2+x_3 \geq 2\\
y_1+y_2+y_3 \geq 2.
\end{align*} 
Since we know that the Butterfly-graph is nonseparable, $d>1$, and since the previous hyperplanes fulfill the 
definition, the level of separability for the Butterfly-graph is $d=2$. 

Analogously, if we link both triangles with a new link $z$, we get the Glasses-graph.  
A slight modification of the hyperplanes shows that the Glasses-graph has level of separability $d=2$:
\begin{align*}
x_1+x_2+x_3+3z \geq 5\\
y_1+y_2+y_3 \geq 2.
\end{align*} 
Observe that we \emph{force} the link $z$ to be operational, adding the term $3z$ in the first hyperplane. 
Finally, let us consider Monma graph $M_{2,2,1}$ from Figure~\ref{ejemploMonma}, where the three paths 
have respective lengths 2, 2 and 1, and the respective links from each path are sequentially identified with the binary states $x_1, x_2, y_1, y_2$ and $z$. 
The reader is invited to check that the level of separability in Monma graph $M_{2,2,1}$ 
 is also $d=2$, and the following pair of hyperplanes works:
\begin{align*}
10x_1+10x_2+y_1+y_2+z \geq 12\\
x_1+x_2+10y_1+10y_2+z \geq 12.
\end{align*} 
Currently, there is no constructive algorithm to produce the minimum number of hyperplanes for an SMBS. 
We wish to develop a complementary theory to the one presented in Section~\ref{all-terminal} for separable graphs, but finding the correct level of separability for any given graph. Inspired by Theorem~\ref{base}, 
we propose the following:
\begin{conjecture}\label{d-separabilidad}
Let $d$ be a fixed positive integer. Then, the all-terminal reliability evaluation of graphs with 
level of separability $d$ belongs to the class $\mathcal{P}$ of polynomial-time problems.
\end{conjecture}

\section{Conclusions}
\label{conclusiones}
In this work, we study the reliability evaluation of stochastic binary systems (SBS) and some properties arising from the definition of separability, and we apply these concepts to take a new look at the all-terminal reliability model. 

An efficient representation of separable systems is presented, and a full characterization of these special systems is obtained for some particular models. The major strength of separable systems is their efficient representation. The major shortcoming is that the reliability evaluation is still ${\mathcal NP}$-hard. 


Separable systems accept polynomial-time reliability evaluation when restricted to 
the all-terminal reliability model. This result was discovered using functional analysis and feasible 
functionals from the links of a graph, meeting separability constraints.  

As future work, we would like to establish Conjecture~\ref{d-separabilidad} for a better understanding 
of nonseparable systems, and the interplay between general SBSs and the all-terminal reliability model, 
which has a wide spectrum of applications.

\section*{Acknowledgment}
We thank Dr. Luis St\'abile for a fruitful discussion shared with him about the separability of Monma graphs. His generosity is much appreciated by the authors. This work was partially supported by Project 395 CSIC I+D \emph{Sistemas Binarios Estoc\'asticos Din\'amicos}, by Project ANII FCE\_1\_2019\_1\_156693 \emph{Teoría y Construcción de Redes de Máxima Confiabilidad} and by Project MATHAMSUD 19-MATH-03 \emph{RareDep}.

\bibliographystyle{plain}
\bibliography{Biblio}
\end{document}